\newif\ifcomments  
\commentsfalse
\documentclass[12pt]{article}

\usepackage{amsmath,amssymb,pifont}
\usepackage{multicol}
\usepackage{amstext}
\usepackage{amsthm}
\usepackage{multirow}
\usepackage{booktabs}
\usepackage[skip=0pt]{subcaption}
\usepackage{times}
\usepackage{lipsum}
\usepackage[shortlabels]{enumitem}
\usepackage{cancel}
\usepackage{wrapfig}
\usepackage{array}
\usepackage{siunitx}
\usepackage{csvsimple}
\usepackage[multidot]{grffile}
\usepackage{bbm}
\usepackage{dblfloatfix}
\usepackage{hyperref}
\usepackage{makecell}
\usepackage{bbm, dsfont}
\usepackage{mathtools}
\usepackage{xcolor}
\usepackage{comment}

\usepackage[multiple]{footmisc}
\usepackage{mathrsfs}
\usepackage{todonotes}
\usepackage{tikz}

\usepackage{algorithm,algorithmic}

\hypersetup{final}

\newcommand{\bfI}{\ensuremath{\mathbf{I}}}

\newcommand{\bfm}{\ensuremath{\mathbf{m}}}

\newcommand{\bfv}{\ensuremath{\mathbf{v}}}

\newcommand{\calA}{\ensuremath{\mathcal{A}}}

\newcommand{\calD}{\ensuremath{\mathcal{D}}}

\newcommand{\calM}{\ensuremath{\mathcal{M}}}
\newcommand{\calN}{\ensuremath{\mathcal{N}}}

\newcommand{\calR}{\ensuremath{\mathcal{R}}}

\newcommand{\bbI}{\ensuremath{\mathbb{I}}}

\newtheorem{lem}{Lemma}

\newtheorem{thm}[lem]{Theorem}

\newtheorem{claim}[lem]{Claim}

\makeatletter
\newcommand{\vast}{\bBigg@{4}}
\newcommand{\Vast}{\bBigg@{5}}
\makeatother

\newcommand{\ltwo}[1]{\left\|#1\right\|_2}

\DeclarePairedDelimiterX{\infdivx}[2]{(}{)}{%
  #1\;\delimsize\|\;#2%
}

\renewcommand{\epsilon}{\varepsilon}

\newcommand{\adv}{{\sf Adv}}

\usepackage{fullpage}
\usepackage[numbers, sort, comma, square]{natbib}

\newcommand{\Fstar}{F^*}

\title{Fully Adaptive Composition for  Gaussian\\ Differential Privacy}
\author{
Adam Smith\thanks{Boston University. \texttt{ads22@bu.edu}}
\and
Abhradeep Thakurta\thanks{Google Research - Brain Team.  \texttt{athakurta@google.com}}
}

\begin{document}
\maketitle

\begin{abstract}
    We show that Gaussian Differential Privacy, a variant of differential privacy tailored to the analysis of Gaussian noise addition, composes gracefully even in the presence of a \textit{fully adaptive} analyst. Such an analyst selects mechanisms (to be run on a sensitive data set) and their privacy budgets adaptively, that is, based on the answers from other mechanisms run previously on the same data set. In the language of Rogers, Roth, Ullman and Vadhan, this gives a filter for GDP with the same parameters as for nonadaptive composition. 
\end{abstract}

Let $\calD$ be a domain of data records, and $\calR$ be a set of potential outputs for a mechanism (without loss of generality, $\{0,1\}^{*}$). We say a randomized function (henceforth, a \textit{mechanism}) $\calM:D^* \to \calR$ satisfies $\mu$-\textit{Gaussian differential privacy} (or $\mu$-GDP)~\cite{GDP-DL} if for every pair of neighboring data sets $D$ and $D'$, the distributions $\calM(D)$ and $\calM(D')$ are ``at least as hard to tell apart'' as two Gaussian random variables with standard deviation 1 and means 0 and $\mu$. Specifically, if $Z\sim \calN(0,1)$, then we ask that that exists a randomized map $F$ (depending on $D$ and $D'$) such that $\calM(D)\sim F(Z)$ and $\calM(D') \sim F(\mu + Z)$. One can also characterize GDP, via Blackwell's theorem, in terms of bounds on the ROC curve for all tests that aim to distinguish $\calM(D)$ from $\calM(D')$~\cite{GDP-DL}; we will not need that characterization here.

We prove that Gaussian Differential Privacy (GDP)~\cite{GDP-DL} composes even when the analyst is fully adaptive; the privacy parameters adds up in exactly the same way as in nonadaptive composition. In the language of \cite{odometers}, there is a \emph{filter} for this privacy concept that simply sums the squares of  privacy budgets and compares them to a threshold.

Consider an interaction between a data curator $\calM_0$ holding a data set $D$ and an adversarial analyst $\adv$. In round $i$, the analyst specifies a privacy budget $\mu_i$  and a mechanism $\calM_i$ that is $\mu_i$-GDP; the curator computes $a_i = \calM_i(D)$ and sends it to the analyst. The analyst's choices at each stage can depend arbitrarily on the interaction up to that point. 

What restriction can the curator impose to ensure that the overall interaction is $\mu_0$-GDP, for some predetermined budget $\mu_0$? If the sequence of $\mu_i$'s were fixed ahead of time, then it would suffice to require that $\sum_i \mu_i^2 \leq \mu)^2$—this is the composition result of \cite{GDP-DL}. We show that this same restriction suffices even when the $\mu_i$'s are selected adaptively.

\begin{thm} Let $\calM_0$ be an interactive mechanism (as above) which ensures that for every analyst, every sequence  of queries  $(\calM_1,\mu_1),...,(\calM_T,\mu_T)$ satisfies \[\sum\limits_{i=1}^n\mu_i^2\leq \mu_0^2\ . \] Then $\calM_0$ satisfies $\mu_0$-GDP.\label{thm:GDP}
\end{thm}

A different proof of this statement was given independently and concurrently by Koskela et al.~\cite[Theorem 11]{koskela2022individual}.

\begin{proof} 
Fix two neighboring data sets $D$ and $D'$. For the purposes of analyzing privacy, we can assume w.l.o.g. that the mechanism receives as input just a bit $b\in \{0,1\}$ specifying which of $D$ and $D'$ to use. A mechanism $\cal A$ is then $\mu$-GDP if there is a randomized function $F:\mathbb{R} \to \calR$ such that $\calA(b) = F(b\cdot \mu + Z)$ where $Z\sim N(0,1)$. 

Now consider the following interaction procedure between a mechanism $\calA$ and an adversary $\adv$: in round $i$, $\adv$ sends $\mu_i$ to $\calA$ and receives $H_i = b\cdot\mu_i+Z_i$, where $Z_i\sim\calN(0,1)$ is a fresh draw from the standard Gaussian distribution. The mechanism continues to answer queries as long as $\sum_i \mu_i^2 \leq \mu_0^2$.
The adversary may select $\mu_i$ based on $H_1,..,H_{i-1}$, but the mechanism applied at each round always consists of additive Gaussian noise. We first show that to prove Theorem~\ref{thm:GDP}, it suffices to show that $\calA$ is $\mu_0$-GDP.

\begin{claim}
If, for all $\adv$, the mechanism $\calA$ defined above satisfies $\mu_0$-GDP, then so is $\calM_0$.
\label{cl:1}
\end{claim}

\begin{proof}
The proof follows from the simulation-based definition of GDP. Given an adversary $\adv'$ that interacts with $\mu_0$ and a pair of data sets $D$ and $D'$, we can construct an adversary $\adv$ that interacts with $\calA$ and simulates the transcript of $\calM_0$'s interaction with $\adv'$: for each query $(\mu_i, \calM_i)$ issued by $\adv'$, the adversary $\adv$ constructs a (possibly randomized) function $F_i$ such that $\calM_i(D)\sim F_i(Z_i)$ and $\calM_i(D') \sim F_i(\mu_i + Z_i)$. It then issues the query $\mu_i$ to $\calA$, receives $H_i$, and feeds $F_i(H_i)$ as a response to $\adv'$. If $\calA$ satisfies $\mu_0$-GDP, then there is a randomized function $\Fstar$ such that $\Fstar(b\cdot \mu_0 + Z_0)$ faithfully simulates the transcript of the interaction of $\calA(b)$ with $\adv$ when $Z_0\sim \calN(0,1)$. The internal messages that $\adv$ feeds to $\adv'$ then correctly simulate the interaction of $\adv'$ with $\calM_0$.
\end{proof}

To show that mechanism $\calA$ indeed satisfies $\mu_0$-GDP, we explicitly specify the postprocessing function $\Fstar$ (depending on $\mu_0$) that takes as input $W_0\leftarrow b\mu_0+Z_0$, where $Z_0\sim\calN(0,1)$ and simulates the responses $\calA(b)$ would give to a sequence of (adaptively chosen) messages $\mu_i$. 

Without loss of generality, we make two simplifying assumptions: First, we assume the sequence of queries satisfies $\sum_{i=1}^T \mu_i^2 \leq \mu_0^2$. This can be enforced by having $\calA$ (and $\Fstar$) refuse to answer queries that would push the sum over the threshold. The analyst/adversary can predict this decision on its own, so it conveys no additional information. Second, we assume that $\mu_0 =1$, since we can replace the actual queries $\mu_i$ with $m_i = \frac{\mu_i}{\mu_0}$ and rescale answers appropriately.

\paragraph{Warm-up: An Online Simulation for Nonadaptive Adversaries.} Suppose (with significant loss of generality!) that the analyst is not adaptive. That is, it decides ahead of time on the sequence $\bfm = (m_1,..,m_T)$, with $\|\bfm\|_2\leq 1$, that it will ask as queries.

On input $W_0\sim \calN(b,1)$ and this sequence,  
$\Fstar$ must compute a sequence of variables $W_i$ such that $W_1,...,W_T$ are jointly distributed as $\sim \calN(b \bfm, \bbI)$. However, we will get an appropriate simulator by imposing an additional requirement: it must  do this computation \textit{online}. That is,  $W_i$ can depend on $(m_1,...,m_i)$ and any internal state, but not on future queries $m_{i+1},...,m_T$.

A natural approach is to respond with $m_i W_0$ at each round $i$.  This sequence  has the correct mean $b \bfm$; unfortunately, its covariance matrix is  $\bfm^\top \bfm$, which is incorrect. One can fix this by adding a component of \textit{independent} random noise $(U_1,...,U_T)$, distributed as $\calN(0, \Sigma_T)$ where $\Sigma_T = \bfI - \bfm^\top\bfm$. Since $\|\bfm\|\leq 1$, the matrix $\Sigma_T$ is PSD and thus a valid covariance matrix for a multivariate Gaussian distribution.  The sequence $W_i = m_i W_0 + U_i$ still has the right mean, and its covariance matrix is now $\bfm^\top \bfm + \Sigma_T = \bbI$, as desired. If the $m_i$ were known ahead of time, we would be done, since the simulator could generate $(U_1,...,U_T)$ before the start of the interaction, and add the $U_i$'s as the interaction proceeds.
The issue is that we have set our selves the more stringent goal of executing the simulation without knowing the $m_i$'s ahead of time. The challenge, therefore, is to sample the $U_i$'s online, as the values $m_i$ are revealed. 

To see why this is possible, recall that the Cholesky decomposition of $\Sigma_T$ gives a lower-triangular matrix $L$ such that $\Sigma_T = L L^\top$. The matrix $L$ is  uniquely determined if we require that $L$ have exactly $rank(\Sigma_T)$ positive diagonal elements and $T-rank(\Sigma_T)$ columns that are all zero. We refer to this $L$ as the \emph{canonical Cholesky factor} of $\Sigma_T$. 

For any $i$ between 1 and $T$, let $L_i$ denote the $i\times i$ submatrix of $L$ restricted to the first $i$ rows and columns. We can write $L$ as four blocks, with $L_i$ on the top left and zeros on the top right: 
\[L= \begin{bmatrix}
L_i & \mathbf{0}_{i \times (T-i)}\\
\vdots & \ddots \\
\end{bmatrix}
\]
The product $L L^\top$ can also be written in block form, where the top left 
block is $L_iL_i^\top$. This must equal the corresponding submatrix of $\Sigma_T$, which we know is $\bbI - \bfm_i \bfm_i^\top$ where $\bfm_i = (m_1,...,m_i)$. That is,
\[L L^\top = \begin{bmatrix}
L_iL_i^\top & \cdots\\
\vdots & \ddots\\
\end{bmatrix} = \Sigma_T = \begin{bmatrix}
\bbI - \bfm_i \bfm_i^\top & \cdots\\
\vdots & \ddots\\
\end{bmatrix} \, .\]

Thus, $L_i$ is the canonical Cholesky factor of $\bbI - \bfm_i \bfm_i^\top$, which is the covariance matrix we need for $U_1,...,U_i$. We can compute $L_i$ given only the queries $(m_1,...m_i)$ from the first $i$ rounds. 

We obtain the following simple algorithm for $\Fstar$, on input $W_0$:
\begin{enumerate}
    \item Generate $V_1,...,V_T$ i.i.d. from $\calN(0,1)$.
    \item For $i=1$ to $T$: 
    \begin{enumerate}
        \item Receive $m_i$ from $\adv$.
        \item Let $\bfm_i = (m_1,...,m_i)$ and let $L_i$ be the canonical Cholesky factor\footnote{Specifically, we can write $L_i=\begin{bmatrix}
                 L_{i-1} & \mathbf{0}\\
                 -m_i\left(L_{i-1}^{-1}\bfm_{i-1}\right)^\top & \frac{1-\ltwo{\bfm_i}^2}{1-\ltwo{\bfm_{i-1}}^2}
                 \end{bmatrix}$.} 
         of $ \Sigma_i = \bbI - \bfm_i \bfm_i^\top$. 
        \item Let $U_i$ be the last entry of $L_i \bfv_i$ where $\bfv_i = (V_1,...,V_i)$. 
        \item Output $W_i = m_i W_0 + U_i$.
    \end{enumerate}
\end{enumerate}

For every fixed $\bfm$ and $b$, if $W_0\sim \calN(b, 1)$, then the joint distribution of the responses $W_1,...,W_T$ is $\calN(b\bfm , \bbI)$. In particular, for every setting of $W_1,..,W_{i-1}$, the random variable $W_i$ is distributed as $\calN(b m_i, 1)$ (which is the distribution of the response sent by $\calA(b)$). This proves the result for the limited case when $\bfm$ is fixed ahead of time by $\adv$.

\paragraph{Fully Adaptive Composition.} We now argue that the same simulation works for fully adaptive queries, where $m_i$ may depend on the interaction so far. Our analysis of the noninteractive setting showed the following:  \emph{for every setting of $b$ and $W_1,..,W_{i-1}$ and the first $i$ queries $m_1,...,m_i$, the random variable $W_i$ is distributed as $\calN(b m_i, 1)$}. Furthermore, it is computed using only the values $W_0$,  $m_1,...,m_i$ and $w_1,...,w_{i-1}$.

Because of the statement holds for all $m_i$ (and not, say, just with high probability), the statement continues to hold when $m_i$ is selected by $\adv$. By induction over the rounds $i$, the transcript of an interaction with $\Fstar(W_0)$ (where $W_0\sim\calN(b,1)$) will be identically distributed to that of an interaction with $\calA(b)$. This concludes the proof of Theorem~\ref{thm:GDP}.
\end{proof}

\section*{Acknowledgements}

We thank Thomas Steinke for helpful discussions. 

\bibliographystyle{alpha}
\bibliography{reference}

\end{document}